\documentclass[letterpaper, journal, 10pt]{IEEEtran}
\usepackage{amsmath,amsfonts,amssymb}
\usepackage{array}
\usepackage[caption=false,font=normalsize,labelfont=sf,textfont=sf]{subfig}
\usepackage{textcomp}
\usepackage{stfloats}
\usepackage{amsmath}
\usepackage{bbm}
\usepackage{url}
\usepackage{verbatim}
\usepackage{graphicx}
\usepackage{cite}
\usepackage{amsthm}
\usepackage{enumitem}
\usepackage{xcolor}
\newtheorem{defn}{Definition}

\newtheorem{thm}{Theorem}
\newtheorem{lem}{Lemma}
\newtheorem{prop}{Proposition}

\usepackage{tabularx}
\usepackage{adjustbox}
\usepackage{comment}
\usepackage{algpseudocode}
\usepackage{hyperref}

\usepackage{etoolbox}
\makeatletter
\patchcmd{\section}
  {3.0ex plus 1.5ex minus 1.5ex}{1.5ex plus 0.4ex minus 0.2ex}
  {}{\PackageError{SATA}{Section spacing patch failed}{}}
\patchcmd{\section}
  {0.7ex plus 1ex minus 0ex}{0.35ex plus 0.15ex minus 0.1ex}
  {}{\PackageError{SATA}{Section spacing patch failed}{}}
\patchcmd{\subsection}
  {3.5ex plus 1.5ex minus 1.5ex}{1.3ex plus 0.4ex minus 0.2ex}
  {}{\PackageError{SATA}{Subsection spacing patch failed}{}}
\patchcmd{\subsection}
  {0.7ex plus .5ex minus 0ex}{0.35ex plus 0.15ex minus 0.1ex}
  {}{\PackageError{SATA}{Subsection spacing patch failed}{}}
\expandafter\patchcmd\csname\string\proof\endcsname
  {\topsep6\p@\@plus6\p@\relax}
  {\topsep3\p@\@plus1\p@\@minus1\p@\relax}
  {}{\PackageError{SATA}{Proof spacing patch failed}{}}
\makeatother

\def\BibTeX{{\rm B\kern-.05em{\sc i\kern-.025em b}\kern-.08em
    T\kern-.1667em\lower.7ex\hbox{E}\kern-.125emX}}

\begin{document}

\title{Self-Adaptive Threshold ALOHA}

\author{Ahsen Topbas, Aimin Li,~\IEEEmembership{Member,~IEEE} and Elif Tugce Ceran%
\vspace{-2em}
\thanks{Authors are with the Department of Electrical and Electronics Engineering, Middle East Technical University (METU), 06800 Ankara, Turkiye (e-mails: \{ahsen.topbas, aimin, elifce\}@metu.edu.tr).}%
\thanks{This work was supported in part by the TÜBİTAK 1515 Frontier Research and Development Laboratories Support Program for
Turk Telekom neXt Generation Technologies Lab (XGeNTT) under Project
5249902 and in part by the European Union through ERC Advanced Grant 101122990--GO SPACE--ERC-2023-A.}
}



\maketitle

\begin{abstract}
We propose Self-Adaptive Threshold ALOHA (SATA), a distributed random-access policy to minimize the Age of Information (AoI). 
SATA uses 1-bit broadcast feedback and requires no explicit coordination or message passing among nodes. 
{We show that this minimal feedback is sufficient for each node to perfectly track the number of active users and adjust its access probabilities accordingly.} 
For any initial network state, SATA converges to a collision-free TDMA steady state in finite time whenever the age threshold satisfies $\Gamma \geq n$, where $n$ is the network size. 
We derive closed-form expressions for the long-term average AoI and throughput in steady state, and establish that, for $\Gamma=n$, the expected transient duration scales as $O(n\log n)$. 
Simulation results confirm that SATA closely approaches TDMA performance across a wide range of network sizes, significantly outperforming Slotted ALOHA, Threshold ALOHA, and 1-persistent Threshold Slotted ALOHA. 
Notably, the performance gap between SATA and the benchmark random-access policies becomes more pronounced as network size grows.
\end{abstract}


\section{Introduction}
{The rapid expansion of wireless sensing, the Internet of Things (IoT), and real-time monitoring applications has increased the demand for communication systems capable of delivering timely, up-to-date information \cite{yates2021jsacsurvey}.} The Age of Information (AoI) metric has emerged as a key measure of the time elapsed since the most recent update \cite{kaul2012real}. As a result, minimizing AoI has become a primary objective in the design of communication networks with timeliness goals. Simultaneously, many large-scale wireless systems employ random access mechanisms due to their simplicity, scalability, and distributed operation. Nevertheless, random access networks are prone to collisions and channel contention, which can substantially degrade information freshness. The resulting design question is how much
coordination is actually needed for distributed random access to
approach the freshness performance of scheduled access.

Age-aware random-access schemes have approached this problem by
controlling which users become active and how they contend for the
channel \cite{atabay2020improving,yavascan2021analysis,bidokhti2022sat,dester2026onePersistent,cavalagli2024reinforcement}. 
\cite{atabay2020improving} proposed LAZY, also known as Threshold ALOHA (TA) \cite{yavascan2021analysis}, and \cite{yavascan2021analysis} presented the steady-state analysis of TA. They found that in the large-network limit, the policy converges to a Slotted ALOHA (SA) with fewer users, and the optimal AoI scales with the network size $n$ as $1.4169n$. The performance of Threshold ALOHA degrades as the network size decreases due to its lack of convergence in small networks. 
{\cite{bidokhti2022sat} proposed the Stabilized Age-based Threshold (SAT) policy, which achieves an age-based thinning, similar to the result of \cite{yavascan2021analysis}, and asymptotic scaling of average age as $\frac{e}{2}n \approx 1.359n$.}
The 1-persistent Threshold Slotted ALOHA (1-persistent TSA) \cite{dester2026onePersistent} built on TA by having nodes with AoI exactly equal to the threshold transmit with probability 1.
With this modification, 1-persistent TSA converges to a collision-free steady state and achieves a long-term average AoI of almost $n$ and a throughput of $n\Gamma^{-1}$ ($n$ and $\Gamma$ are the network size and the age threshold, respectively) for $\Gamma=2n-1$. 
{When $\Gamma=n$, the transient duration of 1-persistent TSA grows to an impractical order of $10^{45}$ time slots.}
\cite{cavalagli2024reinforcement} proposed AoI-Q-ALOHA, which is a distributed learning-based random access policy that achieves performance close to TA.  
The closest prior works to ours are TA \cite{yavascan2021analysis} and 1-persistent TSA \cite{dester2026onePersistent}.
However, both policies use predefined, fixed transmission probabilities for active users, which limits the range of network sizes for which they are feasible. 
As a result, designing distributed and adaptive random-access policies that work for a wide range of network sizes remains an open research area. 

\begin{figure}
    \centering
    \includegraphics[width=\linewidth]{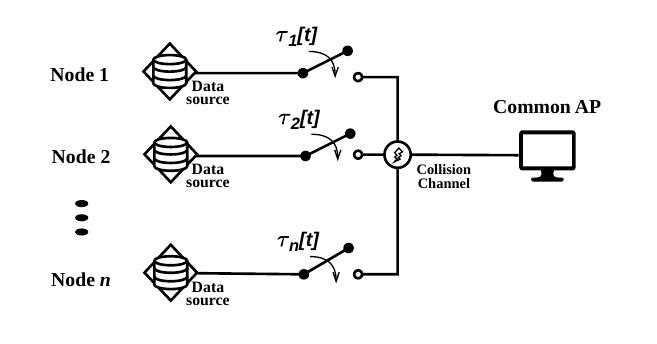}
    \caption{{System model of a random access network of a common access point and $n$ nodes transmitting with adaptive probability $\tau_i[t], i \in \{1,2,..., n \} $.}}
    \label{fig:system_model}
\end{figure}

In this letter, we propose Self-Adaptive Threshold ALOHA (SATA), a distributed random-access policy that uses 1-bit feedback to allow each node to track the network state and adapt its transmission probability in real time. In SATA, nodes with AoI above the threshold transmit with a probability inversely proportional to the estimated number of active users, while a node whose AoI equals the threshold transmits with probability 1, silencing all other active nodes. Through this mechanism, the network organizes itself into a TDMA scheme at steady state in finite time, without any explicit coordination.
The key contributions of this letter are as follows:
\begin{itemize}
    \item {\textbf{Minimal-feedback decentralized adaptation:}}We show that 1-bit feedback is sufficient for each node to perfectly track the number of active users at every time slot and
    adapt its transmission probability accordingly, {without explicit
    coordination or message passing (Lemma~\ref{lem:belief}).}
    \item {\textbf{Finite-time self-organization from random access to TDMA:}} We prove that SATA converges to a collision-free TDMA regime in finite time {almost surely} for any initial network state when $\Gamma \geq n$ (Proposition \ref{prop:mConvergence}). In this regime, each node transmits successfully once
    every $\Gamma$ slots, yielding the closed-form long-term AoI $\overline{\Delta}=\frac{\Gamma+1}{2}$ and long-term throughput $\eta=\frac{n}{\Gamma}$ (Proposition \ref{prop:results}). In particular, $\Gamma=n$ matches collision-free TDMA in the
    long run, achieving average AoI $(n+1)/2$ and unit throughput
\cite{dester2026onePersistent}
    {\item \textbf{Fast transient duration analysis:}} {We establish a finite-time characterization of SATA's self-organization process: for $\Gamma=n$, the expected transient duration scales as $O(n\log n)$ and admits an exponentially decaying tail (Theorem~\ref{thm:nlogn}).} In comparison, 1-persistent TSA has been reported to exhibit prohibitively long transients under the same threshold \cite{dester2026onePersistent}. 
\end{itemize}

\section{System Model}
We consider a random-access network with $n$ communication nodes and a common access point (AP). Time is divided into equal-sized time slots, and the time duration of a time slot is normalized to 1. At the beginning of each time slot, nodes access the AP through a collision channel. 
Each packet transmission takes one time slot. 
If more than one node attempts to transmit, their packets collide and are lost.  
If only one node transmits, it reaches the AP successfully. At the end of each time slot, the AP broadcasts a 1-bit success feedback, {$f\in\{S, F\}$, indicating whether the transmission succeeded.} 

{We adopt the \textit{generate-at-will} model \cite{sun2017update}.} Each node $i$ generates a single information flow, and the Age of Information (AoI) of flow $i$ at time slot $t$ is \cite{kaul2012real}
\begin{equation}
    \delta_i[t] = t - u_i[t],
\end{equation}
where $u_i[t]$ is the packet generation time of the latest received packet of flow $i$ by the AP. The AoI of flow $i$ evolves as
\begin{equation}\label{eq:Disc_AoI_Evolution}
    \delta_i[t] = 
    \begin{cases}
        1, & \text{if node $i$ successfully} \\
        & \text{transmits at time slot $t-1$}, \\
        \delta_i[t-1]+1, & \text{otherwise}.
    \end{cases}
\end{equation}

{The long-term average AoI of the network equals
\begin{equation}
    \bar{\Delta} = \lim_{T \to \infty} \frac{1}{nT} \sum_{t=0}^{T-1} \sum_{i=1}^{n} \delta_i[t],
\end{equation}
if the limit exists. }

The throughput of the random access channel is defined as 
\begin{equation}
    \eta \triangleq \lim_{T \to \infty} \frac{1}{T} \sum_{t=0}^{T-1} \sum_{i=1}^{n} \mathbbm{1} \{\delta_i[t]=1\},
\end{equation}
if the limit exists. Define the network state composed of the ages of nodes as 
    $$\Delta[t] \triangleq \begin{bmatrix}\delta_1[t] & \delta_2[t] & \ldots & \delta_n[t]\end{bmatrix},$$
where the ages take values from $\mathcal{S}=\mathbb{Z}_+^n$ while discrete time $t \in \mathbb{Z}_+$. 
Next, we define the Self-Adaptive Threshold ALOHA (SATA) policy.

\section{SATA}
We consider a threshold-based random access policy in which, if a node transmits successfully, it waits for $\Gamma$ time slots before becoming active with a stationary probability $\tau$ \cite{atabay2020improving}. In such systems, the number of active users, denoted as $m[t]$, is time-varying. 
In Threshold-ALOHA \cite{yavascan2021analysis}, asymptotically, it is shown that $m[t]$ converges in probability to $k_0n$, where $k_0$ is a constant between 0 and 1. Unlike Threshold ALOHA, in SATA, if a node's AoI state exceeds $\Gamma$, it transmits with an adaptive probability $\tau_i[t]\in(0,1]$, {adjusted based on the network state belief}.
Also, a node transmits with probability 1 if its AoI state is equal to $\Gamma$. In that case, all active nodes stay silent, so that the one with $\delta_i[t]=\Gamma$ successfully transmits with probability 1 \cite{dester2026onePersistent}. 

The number of active nodes at time slot $t$, denoted as $m[t]$, is defined as 
\begin{equation}\label{eq:m_def}
    m[t]\triangleq\sum_{i=1}^{n} \mathbbm{1} \{\delta_i[t]>\Gamma\}.
\end{equation}
According to our policy, the adaptive transmission probability for node $i$ at time $t$ is defined as a function of $m[t]$  {
\begin{equation}\label{eqn:tau_i}
    \tau_i[t]=
    \begin{cases}
        1, & \delta_i[t]=\Gamma \\
        p_m, & \delta_i[t]>\Gamma \ \& \ \nexists \delta_j[t]=\Gamma \ \forall j\neq i
        \\
        0, & \text{otherwise},
    \end{cases}
\end{equation}
where $p_m$ is the throughput-optimal choice of transmission probability, that is, the reciprocal of the number of active users, $\frac{1}{m[t]}$, \cite{gallager1992datanets}. Notably, if more than one user has the same inactive state in the initial network state, their states will grow until state $\Gamma$; they transmit with probability 1 and collide. It is a transient condition that lasts at most $\Gamma-1$ time slots.}

\section{Analysis of SATA}
From \eqref{eqn:tau_i}, we truncate AoI states at $\Gamma+1$, which describes $\delta_i[t]>\Gamma$.  
Therefore, AoI processes $\delta_i[t]$ can be described by their truncated versions $A_i^\Gamma[t]$ and it evolves as
\begin{equation}\label{eq:trunc_evol}
    A^{\Gamma}_i[t] = 
    \begin{cases}
        1, & \text{node $i$ successfully} \\
        & \text{transmits at $t-1$}, \\
        \min(A^{\Gamma}_i[t-1]+1,\Gamma+1), & \text{otherwise}.
    \end{cases}
\end{equation} 

The network state can also be defined as the vector of truncated AoI processes: 
$$\mathbf{A}^{\Gamma}[t] = \begin{bmatrix}A^{\Gamma}_1[t] & A^{\Gamma}_2[t] & \ldots & A^{\Gamma}_n[t]\end{bmatrix},$$  
which forms a finite-state Markov Chain (FSMC) with state space $\mathcal{S}^{\Gamma} = \{1, 2, \ldots, \Gamma+1\}^n$. As the network is symmetric, all active nodes have the same successful transmission probability $p_s[t]$ at time slot $t$, given that there does not exist a node with state $\Gamma$. Accordingly, the successful transmission probability of an active node can be written as
{\begin{equation}\label{eq:succ_xmit}
    p_s[t] = p_m(1-p_m)^{m[t]-1},
\end{equation}
where $p_m$ denotes the transmission probability of active nodes.} 
Using \eqref{eq:succ_xmit} and \eqref{eq:m_def}, the transition probability matrix $P^{\Gamma}$ of the FSMC $\mathbf{A}^{\Gamma}[t]$ can be constructed.

The probability that exactly one active node succeeds in a contention slot with $m$ active users is
{ 
\begin{equation}\label{eq:q_m}
    q_m=m p_m(1-p_m)^{m-1}>0.
\end{equation}
}
where SATA sets $p_m=1/m$, with $q_1=1$ for the single-user case.

    \subsection{Network State Belief} 
    
    Each node keeps a network state \textit{belief}, denoted as $\mathbf{B}^{\Gamma}[t]$. After receiving the 1-bit feedback at the end of each time slot, nodes update their beliefs recursively. Since SATA sets the access probability according to the current number of active users, decentralized implementation
requires every node to determine $m[t]$ without explicit state
exchange. The following lemma establishes that the belief state $\mathbf{B}^{\Gamma}[t]$ is sufficient to
determine the number of active users at every slot.
    \begin{lem}\label{lem:belief}
        {In a random-access network with $n$ nodes under SATA, 1-bit feedback is sufficient to determine the number of active users, provided all nodes know the initial network state.}    
    \end{lem}

   \begin{proof}
    See Appendix~\ref{app:belief}.
\end{proof}

    Lemma \ref{lem:belief} states that starting from any known initial state $\mathbf{A}^{\Gamma}[0]$, nodes can uniquely determine the state transition using the 1-bit feedback at every time slot $t$, since the AoI evolution can be uniquely inferred from the 1-bit feedback as in \eqref{eq:trunc_evol}. 
    This enables decentralized adaptation of the transmission probability based on the current number of active users, $m[t]$. Also, nodes can determine whether there is a node with state $\Gamma$ and remain silent if so. 

    \begin{defn}
        The absorbing set can be defined as: 
        $$\mathcal{D}=\{\mathbf{A^\Gamma}\in \mathcal{S}^\Gamma:A_i^\Gamma\in \{1,2,...,\Gamma\},A_i^\Gamma\neq A_j^\Gamma, \forall i\neq j \}.$$
    \end{defn}
 
 The set $\mathcal D$ consists of states in which all nodes have
distinct ages below or at the threshold. Consequently, at most one
node reaches $\Gamma$ in any slot, so transmissions within
$\mathcal D$ are collision-free. The next proposition shows that
SATA reaches this set from any initial state in finite time and
remains there thereafter.

    \begin{prop}\label{prop:mConvergence}
        Let $\Gamma \geq n$ and let $\mathbf{A^\Gamma}[0]\in\mathcal{S}^\Gamma$ be any initial state. The target set can be defined as 
        $$\mathcal{D}=\{\mathbf{A^\Gamma}\in \mathcal{S}^\Gamma:A_i^\Gamma\in \{1,2,...,\Gamma\},A_i^\Gamma\neq A_j^\Gamma, \forall i\neq j \},$$
        and let
        $$\tau_\mathcal{D}\triangleq\inf\{t\geq 0 : \mathbf{A^\Gamma}[t]\in \mathcal{D}\}$$
        denote the first hitting time of $\mathcal{D}$ by the network state. Then for any initial state $\mathbf{A^\Gamma}[0]\in\mathcal{S}^\Gamma$, the process reaches $\mathcal{D}$ in finite time almost surely:
        $$\Pr(\tau_\mathcal{D}<\infty)=1.$$
    \end{prop}
    \begin{proof}
        See Appendix~\ref{app:convergence}.
    \end{proof}
    \color{black}

Proposition~\ref{prop:mConvergence} therefore reduces the long-run
behavior of SATA to a periodic collision-free regime. Once
$\mathcal D$ is reached, each node succeeds exactly once every
$\Gamma$ slots. This periodic structure immediately determines the
long-term AoI and throughput.

    \begin{prop}\label{prop:results}
        For $\Gamma \geq n$, the long-term average AoI and throughput of SATA are:
        $$\overline{\Delta}=\frac{\Gamma+1}{2}, \ \ \ \eta=\frac{n}{\Gamma}.$$
    \end{prop}
    \begin{proof}
    See Appendix~\ref{app:results}.
\end{proof}

In particular, setting $\Gamma=n$ gives
\[
\overline{\Delta}=\frac{n+1}{2},
\qquad
\eta=1,
\]
matching collision-free TDMA in the long run. Long-run performance alone, however, does not indicate how quickly
this regime is reached. Over a finite operating horizon, a long
transient can contribute substantially to the realized AoI.
We therefore next characterize the time required for SATA to enter
$\mathcal D$.

\subsection{Transient Analysis}
When a node transmits successfully, its AoI resets to 1, and it remains silent until its age reaches $\Gamma$. Upon reaching $\Gamma$, it deterministically transmits with probability 1, and this periodic transmission pattern repeats every $\Gamma$ time slots for the rest of the network operation.
As a result, the transient period consists of two types of time slots: contention slots, during which the $m$ active nodes compete for channel access, and periodic slots, during which previously successful nodes transmit deterministically. 
The coexistence of contention and periodic slots leads to an irregular transmission pattern, motivating the derivation of the expected transient duration, $E[\tau_\mathcal{D}]$. 
{To this end, we assume the threshold equals the network size, $\Gamma=n$.}

\begin{thm}\label{thm:nlogn}
    Under $\Gamma=n$ and $p_m=m^{-1}$, the hitting time $\tau_\mathcal{D}$ has an exponentially decaying tail, and consequently, it scales with the network size as 
    $$\mathbb{E}[\tau_\mathcal{D}]=O(n\log n),$$
    and the upper bound has the form
    $$\mathbb{E}[\tau_\mathcal{D}] \lesssim 2.18 n \log n.$$
\end{thm}
\begin{proof}
    See Appendix~\ref{app:transient}.
\end{proof}

\section{Simulation Results}
\begin{figure}
    \centering
    \includegraphics[width=0.78\linewidth]{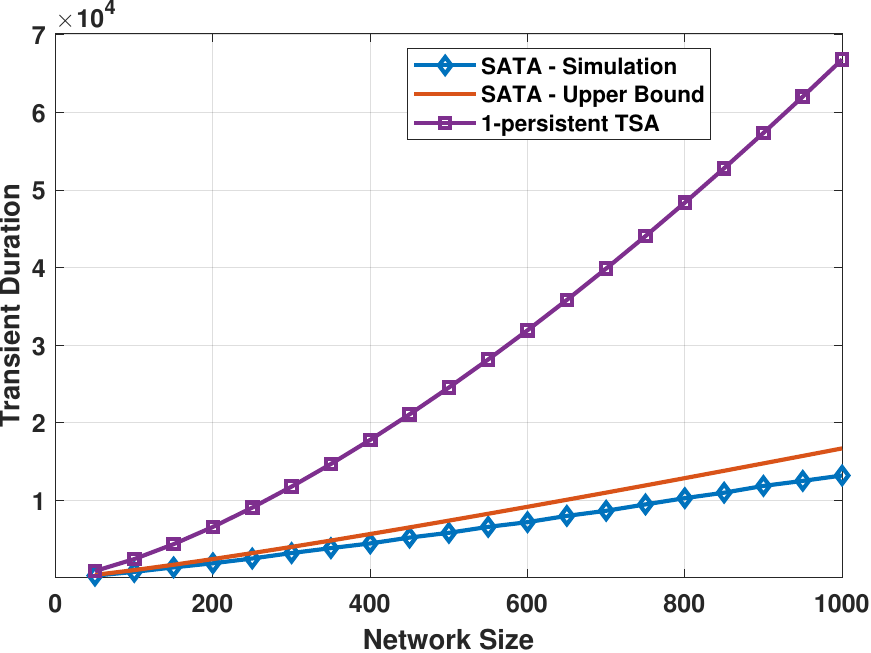}
    \caption{Transient duration of SATA as the network size varies from 50 to 1000 {for $10^2$ iterations of $10^5$ time slots.}\vspace{-0.5cm}}
    \label{fig:transient}
\end{figure}

{In Fig. \ref{fig:transient}, the transient durations of SATA and 1-persistent TSA \cite{dester2026onePersistent} are shown. For 1-persistent TSA, we used their optimal parameters, $\Gamma=2n-1$ and $\tau=2.5/n$, and for SATA, $\Gamma=n$ and $\tau$ is adaptive as in \eqref{eqn:tau_i}. 
In \cite{dester2026onePersistent}, it is reported that the transient duration of 1-persistent TSA is $10^{45}$ for $n=100$ and $\Gamma=n$, for which the transient duration of SATA is approximately $1000$.
As $n$ increases, the upper bound of SATA becomes looser, yet it still yields finite durations. The transient performance of SATA is superior to 1-persistent TSA. Notably, SATA converges to a steady-state TDMA in a reasonable finite time for network sizes between 50 and 1000, while 1-persistent TSA converges to a collision-free steady-state with cycle length $2n-1$. }

\begin{figure}
    \centering
    \includegraphics[width=0.78\linewidth]{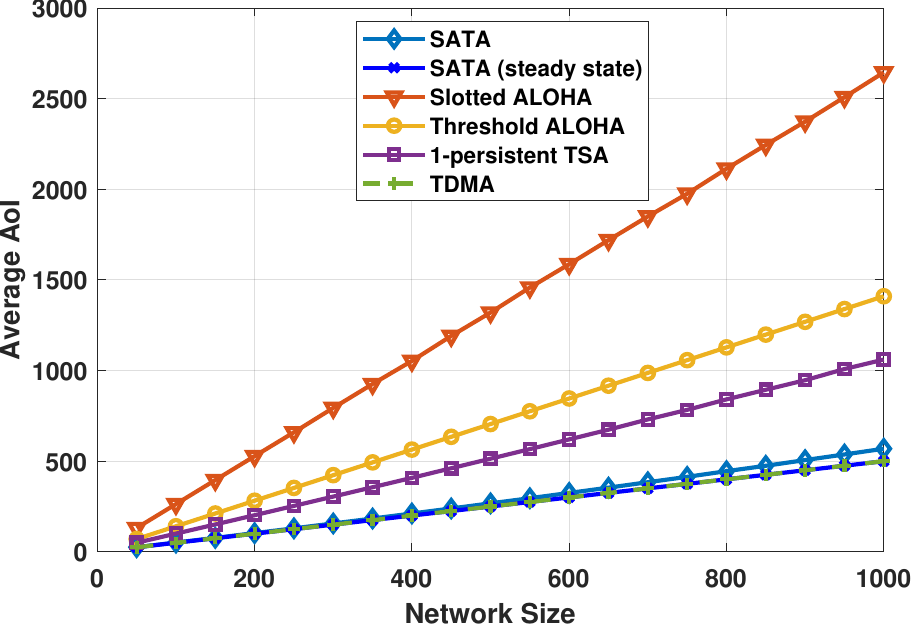}
    \caption{Average AoI as the network size varies from 50 to 1000 {for $10^2$ iterations of $10^5$ time slots.}\vspace{-0.5cm}}
    \label{fig:aoi}
\end{figure}

Fig. \ref{fig:aoi} presents the performance comparison of Slotted ALOHA, Threshold ALOHA (TA) \cite{yavascan2021analysis}, 1-persistent TSA \cite{dester2026onePersistent}, TDMA, and SATA. 
We have performed Monte Carlo simulations for Fig. \ref{fig:aoi} except TA, whose theoretical result is presented. In small network sizes (e.g., between 50 and 200), TA might not reach its theoretical limit, which is based on asymptotic analysis. In this case, SATA outperforms even TA's best performance.  
1-persistent TSA and TA are evaluated with their optimal parameters. 
As $n$ increases, SATA's performance gain becomes more pronounced relative to the benchmark policies. 
Notably, SATA works for a wide range of network sizes, as it is an adaptive policy with non-asymptotic analysis. 
SATA is the only policy that achieves performance closest to TDMA. The small difference between the results for SATA and TDMA is due to SATA's transient duration. Notably, at steady state, the performance of SATA is identical to TDMA.  

\section{Conclusion}
We proposed Self-Adaptive Threshold ALOHA (SATA), a distributed random-access policy to minimize the Age of Information (AoI). Using only 1-bit channel feedback, nodes can track the network state and adapt their transmission probabilities without explicit coordination. We showed that SATA converges to a collision-free TDMA regime in finite time for $\Gamma \ge n$, derived closed-form expressions for the steady-state average AoI and throughput, and established bounds on the expected transient duration. Simulation results show that SATA achieves performance close to TDMA and significantly outperforms existing threshold-based random-access policies, with gains increasing as network size grows. Future work includes extending the framework to imperfect feedback, heterogeneous networks, and different network goals beyond information freshness.

\bibliographystyle{IEEEtran}
\bibliography{main}

@ARTICLE{yates2021jsacsurvey,
  author={Yates, Roy D. and Sun, Yin and Brown, D. Richard and Kaul, Sanjit K. and Modiano, Eytan and Ulukus, Sennur},
  journal={IEEE Journal on Selected Areas in Communications}, 
  title={Age of Information: An Introduction and Survey}, 
  year={2021},
  volume={39},
  number={5},
  pages={1183-1210},
  doi={10.1109/JSAC.2021.3065072}}

@ARTICLE{dester2026onePersistent,
  author={Dester, Plínio S. and Souza, Richard Demo and Cardieri, Paulo},
  journal={IEEE Communications Letters}, 
  title={Analysis of the 1-Persistent Age Threshold Slotted ALOHA}, 
  year={2026},
  volume={30},
  number={},
  pages={1066-1070},
  doi={10.1109/LCOMM.2026.3662201}}

@inproceedings{cavalagli2024reinforcement,
  title={Reinforcement Learning for Age of Information Aware Transmission Policies in Slotted ALOHA Channels},
  author={Cavalagli, Chiara and Badia, Leonardo and Munari, Andrea},
  booktitle={IEEE ISWCS},
  pages={1--6},
  year={2024},
}

@article{yavascan2021analysis,
  title={Analysis of slotted ALOHA with an age threshold},
  author={Yavascan, Orhan Tahir and Uysal, Elif},
  journal={IEEE J. Sel. Areas Commun.},
  volume={},
  number={},
  pages={},
  year={2021},
}

@ARTICLE{bidokhti2022sat,
  author={Chen, Xingran and Gatsis, Konstantinos and Hassani, Hamed and Bidokhti, Shirin Saeedi},
  journal={IEEE Transactions on Information Theory}, 
  title={Age of Information in Random Access Channels}, 
  year={2022},
  volume={68},
  number={10},
  pages={6548-6568},
  doi={10.1109/TIT.2022.3180965}}

@article{sun2017update,
  title={Update or wait: How to keep your data fresh},
  author={Sun, Yin and Uysal-Biyikoglu, Elif and Yates, Roy D and Koksal, C Emre and Shroff, Ness B},
  journal={IEEE Trans. Inf. Theory},
  volume={63},
  number={11},
  pages={7492--7508},
  year={2017},
}

@inproceedings{kaul2012real,
  title={Real-time status: How often should one update?},
  author={Kaul, Sanjit and Yates, Roy and Gruteser, Marco},
  booktitle={IEEE INFOCOM},
  pages={},
  year={2012},
  organization={}
}

@inproceedings{atabay2020improving,
  title={Improving age of information in random access channels},
  author={Atabay, D. C. and Uysal, E. and Kaya, O.},
  booktitle={IEEE INFOCOM WKSHPS},
  pages={},
  year={2020},
}

@book{gallager1992datanets,
  title={Data Networks},
  author={Bertsekas, Dimitri P. and Gallager, Robert G.},
  edition={2nd},
  year={1992},
  publisher={Athena Scientific},
  address={Belmont, MA, USA},
  isbn={978-1-886529-22-9}
}

\appendices

\section{Proof of Lemma~\ref{lem:belief}}\label{app:belief}
\begin{proof}
Suppose that $\mathbf B^\Gamma[t]$ is known at the beginning of
slot $t$. Since $\mathbf B^\Gamma[t]$ contains the truncated AoI
values, each node can determine the number of active users $m[t]$
and the number of nodes with state $\Gamma$.

We consider three cases.

\emph{Case 1: No node is at state $\Gamma$.}
Only the $m[t]$ active nodes may transmit.
If the feedback is $S$, exactly one active node succeeds, so one
state $\Gamma+1$ resets to $1$.
If the feedback is $F$, no active node succeeds.
In either case, all remaining states evolve deterministically
according to~\eqref{eq:trunc_evol}.

\emph{Case 2: Exactly one node is at state $\Gamma$.}
That node transmits with probability one while all active nodes
remain silent. Hence the transmission succeeds, the state-$\Gamma$
node resets to $1$, and all remaining states evolve deterministically.

\emph{Case 3: At least two nodes are at state $\Gamma$.}
All such nodes transmit simultaneously and collide, while the active
nodes remain silent. Therefore, every state-$\Gamma$ node moves to
the active state $\Gamma+1$, and the remaining states evolve
deterministically.

Thus, in all cases, $\mathbf B^\Gamma[t+1]$ is uniquely determined
by $\mathbf B^\Gamma[t]$ and the 1-bit feedback. Starting from the
known initial state
$\mathbf B^\Gamma[0]=\mathbf A^\Gamma[0]$, induction yields exact
tracking of $\mathbf B^\Gamma[t]$, and hence of $m[t]$, for all $t$.
\end{proof}

\section{Proof of Proposition~\ref{prop:mConvergence}}\label{app:convergence}
{ 
\begin{proof}
        Consider two nodes in the same inactive state, $A_i^\Gamma[t]= A_j^\Gamma[t]<\Gamma$. Their ages increase together until they reach $\Gamma$, where they transmit simultaneously and cause a collision. Both then move to the active state, $\Gamma+1$, eliminating the duplicate inactive state. If an active node transmits successfully, its age resets to 1 while inactive nodes' ages increase by 1, so no new duplicate inactive states are created.

        The number of active nodes is denoted by $m[t]$ \eqref{eq:m_def}. Let $m[t]=m>0$.
        At a contention opportunity, the probability of a successful transmission is $q_m>0$, as defined in \eqref{eq:q_m}.

        $C_m$ denotes the number of contention slots with $m$ active users, i.e., the waiting time until the next successful transmission. $C_m$ is geometrically distributed since the transmissions in each time slot are independently and identically distributed Bernoulli trials with probability $q_m$.
        The probability of having no successful transmission of active nodes in the first $k$ contention slots is given by 
        \begin{equation}
            \Pr(C_m>k) = (1-q_m)^k, \ k=0,1,\ldots
        \end{equation}

        Since the truncated state space is finite, there exists $\epsilon>0$ such that $q_m \geq \epsilon$ for every contention state with $m$ active nodes. Therefore,
        $$\Pr(C_m>k) \leq (1-\epsilon)^k.$$
        Furthermore,
        $$\lim_{k\rightarrow\infty}(1-\epsilon)^k=0,$$
        which means $C_m$ is finite with probability 1, $\Pr(C_m<\infty)=1$. 
        If we repeat this argument for every $m\in\{1,\ldots,n\}$, $m[t]$ reaches 0 in finite time almost surely. At that time, all node states belong to $\{1,\ldots,\Gamma\}$ and pairwise distinct, so $\mathbf{A^\Gamma}[t]\in \mathcal{D}.$ Hence,
        $$\Pr(\tau_\mathcal{D}<\infty)=1.$$

        Finally, if $\mathbf{A^\Gamma}[t]\in \mathcal{D}$, at most one node can have state $\Gamma$. If such a node exists, it transmits successfully and is reset to 1, while all other states increase by 1; otherwise, all states increase by 1. In either case, the states remain pairwise distinct and within $\{1,\ldots,\Gamma\}$. Therefore, 
        $$\mathbf{A^\Gamma}[t]\in \mathcal{D} \Rightarrow \mathbf{A^\Gamma}[t+1]\in \mathcal{D}.$$ 
        Hence, $\mathcal{D}$ is absorbing. In fact, once the process enters $\mathcal{D}$, the state trajectory evolves periodically within $\mathcal{D}$ rather than converging to a fixed state.
    \end{proof}
}

\section{Proof of Proposition~\ref{prop:results}}\label{app:results}
\begin{proof}
    Proposition \ref{prop:mConvergence} establishes the existence of a steady state, in which all $n$ nodes operate in a TDMA fashion, each transmitting successfully exactly once every $\Gamma$ time slots. Consider node $i$ transmitting successfully at time $t = k\Gamma + c$ for some constants $k \in \mathbb{Z}^+$ and $0 \leq c \leq \Gamma - 1$. Following this transmission, the AoI of node $i$ is reset to 1 and increases by 1 each subsequent time slot until its next successful transmission at $t = (k+1)\Gamma + c$. Therefore, over one complete period of $\Gamma$ slots, the AoI of node $i$ takes values $1, 2, \ldots, \Gamma$, giving a long-term average AoI of $\frac{\Gamma+1}{2}$.
    
    As the condition $\Gamma \geq n$ ensures that there are $n$ periodic transmissions within each period without collision, exactly $n$ successful transmissions occur every $\Gamma$ slots, yielding a throughput of $\frac{n}{\Gamma}.$
\end{proof}

\section{Proof of Theorem~\ref{thm:nlogn}}\label{app:transient}
\begin{proof}
By the duplicate-elimination argument in the proof of
Proposition~\ref{prop:mConvergence}, after an initial period
$T_0\leq n$, all inactive states are pairwise distinct.
Hence, no further threshold collisions occur, and the number
of active users is non-increasing.

Partition the subsequent evolution into blocks of $n$ slots,
and let $M_k$ denote the number of active users at the beginning
of the $k$-th block. Since $\Gamma=n$, each of the $n-M_k$
inactive users occupies exactly one deterministic transmission
slot within the block. Moreover, an active user that succeeds
during the block cannot reach the threshold again within the
same block. Thus, the block contains $M_k$ contention
opportunities.

Whenever $m\geq2$ active users remain, SATA uses $p_m=1/m$,
and the probability of exactly one successful active
transmission is
\begin{equation}
q_m
=
m p_m(1-p_m)^{m-1}
=
\left(1-\frac{1}{m}\right)^{m-1}
\geq \frac{1}{e},
\end{equation}
with $q_1=1$.
Let $S_k$ denote the number of active-user successes during
block $k$. At every contention opportunity before all active
users have succeeded, the conditional success probability is
at least $e^{-1}$; if all active users succeed earlier, then
$S_k=M_k$. Therefore,
\begin{equation}
\mathbb{E}[S_k\mid M_k]
\geq
\frac{M_k}{e}.
\end{equation}
Since no new active users are created after $T_0$,
\[
M_{k+1}=M_k-S_k,
\]
and hence
\begin{equation}
\mathbb{E}[M_{k+1}\mid M_k]
\leq
\left(1-\frac{1}{e}\right)M_k.
\end{equation}
Taking expectations and iterating yields
\begin{equation}
\mathbb{E}[M_k]
\leq
n\left(1-\frac{1}{e}\right)^k.
\end{equation}

Since the transient has not ended after $k$ blocks only if
$M_k\geq1$, Markov's inequality gives
\begin{align}
\Pr(\tau_{\mathcal D}>T_0+kn)
&\leq \Pr(M_k\geq1) \\
&\leq n\left(1-\frac{1}{e}\right)^k.
\end{align}
Thus, $\tau_{\mathcal D}$ has an exponentially decaying tail.

To bound its expectation, define
\[
\alpha\triangleq
-\log\left(1-\frac{1}{e}\right),
\qquad
k_0\triangleq
\left\lceil\frac{\log n}{\alpha}\right\rceil .
\]
Grouping the tail probability over blocks of $n$ slots gives
\begin{align}
\mathbb{E}[\tau_{\mathcal D}]
&\leq
T_0+
n\sum_{k=0}^{\infty}
\Pr(\tau_{\mathcal D}>T_0+kn) \\
&\leq
T_0+n\left(k_0+e\right) \\
&=
\frac{n\log n}{\alpha}+O(n)
=
O(n\log n).
\end{align}
Since $1/\alpha\simeq2.18$, the leading-order upper bound is
$2.18\,n\log n$, completing the proof.
\end{proof}

\end{document}